\documentclass{article}
\usepackage{ijcai25}

\usepackage{times}
\usepackage{soul}
\usepackage[hyphens]{url}
\usepackage[hidelinks]{hyperref}
\usepackage[utf8]{inputenc}
\usepackage[small]{caption}
\usepackage{graphicx}
\usepackage{amsmath}
\usepackage{amsthm}
\usepackage{booktabs}
\usepackage{algorithm}
\usepackage{algorithmic}
\let\oldciteyear\citeyear
\renewcommand{\citeyear}[1]{[\oldciteyear{#1}]}

\newtheorem{theorem}{Theorem}

\title{Responsibility Gap in Collective Decision Making}

\author{
Pavel Naumov$^1$
\and
Jia Tao$^2$
\affiliations
$^1$University of Southampton, United Kingdom\\
$^2$Lafayette College, United States\\
\emails
p.naumov@soton.ac.uk,
taoj@lafayette.edu
}

\usepackage{subcaption}
\newtheorem{definition}{Definition}
\usepackage{amssymb}
\newtheorem{lemma}{Lemma}
\newtheorem{claim}{Claim}
\newenvironment{proof-of-claim}{\begin{trivlist}\item\noindent{\em Proof of Claim.}}{\hfill {\small $\boxtimes$}\end{trivlist}}
\newcommand{\citet}[1]{\citeauthor{#1}~\citeyear{#1}}
\usepackage{xcolor}

\begin{document}

\maketitle

\begin{abstract}
The responsibility gap is a set of outcomes of a collective decision-making mechanism in which no single agent is individually responsible. In general, when designing a decision-making process, it is desirable to minimise the gap. 

The paper proposes a concept of an elected dictatorship. It shows that, in a perfect information setting, the gap is empty if and only if the mechanism is an elected dictatorship. It also proves that in an imperfect information setting, the class of gap-free mechanisms is positioned strictly between two variations of the class of elected dictatorships.
\end{abstract}


\section{Introduction}

AI agents are involved in making significant decisions in our everyday lives -- from driving autonomous vehicles and investing in stock to estimating (in the role of war robots) potential civilian casualties. For such decisions to be socially acceptable, there should be at least one agent responsible for the outcome of the decision. That is, the decision-making mechanism should be designed without responsibility gaps.

The term {\em responsibility gap} (or responsibility void) is used in the literature in two distinct but related ways. First, it refers to situations where an agent who would normally be held responsible lacks {\em moral agency}--for example, minors, animals, and often AI systems~\cite{m04eit,ct15pt,bhlmmp20ai,c20see,g20eit,sm21pt,t21pt,k22eit,o23pt,hv23synthese}. Second, it describes cases where the design of a collective decision-making mechanism is such that no {\em single} agent (artificial or otherwise) can be held accountable for the outcome of the group decision~\cite{bh11pq,d18pss,l21pt,d22,dy23as,sn25jpl}. 

In this paper, we investigate the properties of responsibility gaps--specifically in the second sense--using the concept of an elected dictatorship that we introduce. We demonstrate that in the perfect information setting, the responsibility gap is empty if and only if the mechanism constitutes an elected dictatorship. In the case of imperfect information, we further show that the class of gap-free mechanisms lies strictly between two variants of elected dictatorships. 
After the statement of Theorem~\ref{big result perfect info}, we compare our findings with two closely related results from the literature.


\section{Decision-making Mechanisms}

An example of a collective decision-making mechanism is the Two-person Rule used to launch American Minuteman II intercontinental ballistic missiles with nuclear warheads. Only the President of the United States can authorise the launch of the missiles. Once the President issues a launch order, the crew on the missile launch site must strap to their chairs (in case of a nuclear attack on the launch facility). Then, two on-duty officers must {\em simultaneously} turn their keys to activate the launch. No single officer can turn both keys because they are 12 feet apart~\cite{usaf24}. 

\begin{figure}[ht]
\begin{center}
\scalebox{0.49}{\includegraphics{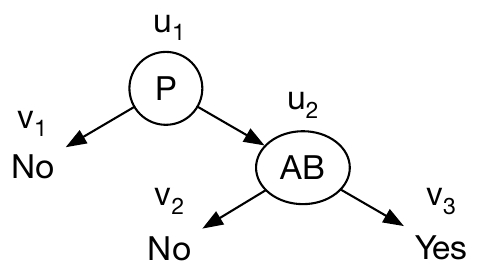}}
\caption{Two-person Rule mechanism.}\label{minuteman figure}
\end{center}
\end{figure}

Figure~\ref{minuteman figure} depicts this mechanism as a tree. 
In this figure, the President ($P$) can unilaterally decide not to launch the missiles and, thus, transition the decision-making process from node $u_1$ to node $v_1$. Alternatively, the President can transition the mechanism to node $u_2$. In that node, officers $A$ and $B$ must simultaneously turn the keys in order for the decision path to end in node $v_3$ where the missiles are automatically launched\footnote{Technically, for missiles to leave the silos, another pair of keys must be turned at another launch control facility~\cite{usaf24}. However, this does not change the responsibility analysis in this paper.}. If either of them does not turn the key, the decision-making process transitions from node $u_2$ to node $v_2$ and the missiles are not launched. 


The next definition generalises the mechanism depicted in Figure~\ref{minuteman figure}.

\begin{definition}\label{mechanism}
A tuple $(V,E,\mathcal{A},\Delta,\tau,\ell)$ is a decision-making mechanism, where
\begin{enumerate}
    \item $(V,E)$ is a rooted directed tree; by $L$ and $D$ we denote the set of all leaf and decision (non-leaf) nodes of this tree, respectively. For each decision node $v\in D$, by $Ch_v$ we denote the set $\{u\in V\mid (v,u)\in E\}$ of children of node $v$,
    \item $\mathcal{A}$ is a set of ``agents'',
    \item $\Delta_v^a$ is a nonempty set of ``actions'' available to an agent $a\in\mathcal{A}$ at a decision node $v\in D$,
    \item $\tau_v:\prod_{a\in\mathcal{A}}\Delta_v^a\to Ch_v$ is a choice function for each decision node $v\in D$,
    \item $\ell:L\to \{\text{Yes},\text{No}\}$ is a labelling function that maps leaf nodes into  ``outcomes''. 
\end{enumerate}
\end{definition}
Although trees are usually assumed to be finite, the results of this paper also hold for a rooted directed tree $(V,E)$ of infinite width (but not infinite depth).

In the case of the mechanism in Figure~\ref{minuteman figure}, nodes $u_1$ and $u_2$ are the decision nodes and nodes $v_1$, $v_2$, and $v_3$ are the leaf nodes. Relation $E$ is represented by the directed edges in the figure. In this example, the set of agents consists of the President $P$ and officers $A$ and $B$. Each of the agents has two actions: ``go left'' (Left) and ``go right'' (Right). 

The choice function $\tau_v$ determines which node the decision process transitions to based on the actions of agents taken at node $v$. 
Intuitively, we assume that only the President decides at node $u_1$ in Figure~\ref{minuteman figure} and only the two officers contribute to the decision at node $u_2$. This can be captured in the more general setting of Definition~\ref{mechanism} by assuming that all three agents act at each decision node, but some actions might not influence the decision at all. For example, function $\tau_{u_1}$ formally takes a tuple (representing actions of agents $P$, $A$, and $B$). However, the value of this function is completely determined by the action of agent $P$ alone. Similarly, the value of $\tau_{u_2}$ is completely determined by the actions of agents $A$ and $B$. 

Labelling function $\ell$ specifies the outcome of the decision-making process at each of the leaf nodes. Note that in this paper, we only consider the mechanisms that make binary ({\em Yes/No}) decisions. {We briefly discuss a more general class of mechanisms in the conclusion}.

Note that each element $\delta$ of the set $\Pi_{a\in\mathcal{A}}\Delta^a_v$ specifies a possible combination of actions of all agents at a decision node $v$. We refer to such a combination $\delta$ as an {\em action profile} at node $v$. By $\delta_a$ we denote the action of agent $a\in\mathcal{A}$ under the profile $\delta$.
Intuitively, by $Next^a_d(v)$ we denote the set of all children of node $v$ to which the decision-making process can transition from node $v$ if agent $a$ chooses action $d$.

\begin{definition}\label{next}
$Next^a_d(v)=\{\tau_v(\delta)\mid \delta_a=d\}$.    
\end{definition}

In the case of our running example, 
$Next^A_{\text{\em Left}}(u_2)=\{v_2\}$ and
$Next^A_{\text{\em Right}}(u_2)=\{v_2,v_3\}$.

\section{Counterfactual Responsibility}

Imagine now a situation, when the President decides to authorise a nuclear strike, the two officers turn the keys, and half of the world is destroyed. Who is responsible for this? The notion of responsibility has been extensively studied in philosophy and law. In philosophy, one of the most commonly discussed approaches~\cite{w17} is to define responsibility based on Frankfurt's principle of alternative possibilities:
``{\em 
\dots\ a person is morally responsible for what he has done only if he could have done otherwise}''
\cite{f69tjop}. In recent works in AI, ``could have done otherwise'' has been interpreted as having a {\em strategy}, at some point during the decision process, to prevent the outcome~\cite{ydjal19aamas,nt19aaai,nt20aaai,bfm21ijcai,sn25jpl}. In this paper, we use the term {\em counterfactual responsibility} to refer to the definition of responsibility based on Frankfurt’s principle. We often omit ``counterfactual'' because this is the only type of responsibility that we consider in this work. 

In our example, at the leaf node $v_3$, all three agents are counterfactually responsible for the decision to launch the missiles because each of them has had a strategy to prevent it. The President could have chosen not to authorise a nuclear strike by transitioning the decision-making process from node $u_1$ to node $v_1$. Each officer could have chosen not to turn the key, unilaterally transitioning the process from node $u_2$ into node~$v_2$.

Because counterfactual responsibility is defined through having a strategy, before formally defining responsibility in an arbitrary decision-making mechanism, we need to define what we mean by ``having a strategy'' to prevent an outcome. Since our decision mechanisms have only two outcomes, {\em Yes} and {\em No}, preventing one of them is equivalent to achieving the other. Below, we use backward induction to formally define the set $win_a(o)$ of all nodes from which an agent $a$ has a strategy to achieve an outcome $o$.

\begin{definition}\label{win definition}
For any outcome $o\in  \{\text{Yes},\text{No}\}$ and any agent $a\in\mathcal{A}$, let set 
$win_a(o)$ be the smallest subset of $V$ such that
\begin{enumerate}
    \item $\ell^{-1}(o)\subseteq win_a(o)$,
    \item if $Next^a_d(v)\subseteq win_a(o)$, then $v\in win_a(o)$, for each decision node $v\in D$ and each action $d\in\Delta^a_v$.
\end{enumerate}
\end{definition}
For example, for the decision mechanism depicted in Figure~\ref{minuteman figure}, we have $win_P(\text{\em No})=\{u_1,v_1,v_2\}$, 
$win_P(\text{\em Yes})=\{v_3\}$,
$win_A(\text{\em No})=\{u_1,v_1,u_2,v_2\}$, 
$win_A(\text{\em Yes})=\{v_3\}$.

For any outcome $o\in\{\text{\em Yes},\text{\em No}\}$, by $\overline{o}$ we mean the other outcome. For example, $\overline{\text{\em Yes}}=\text{\em No}$.

\begin{definition}
A {\bf\em  decision path} is any sequence of nodes $v_1,\dots,v_k$ such that $k\ge 1$ and $v_{i+1}\in Ch_{v_i}$ for each $i<k$.    
\end{definition}

The next definition formally captures Frankfurt’s principle of 
alternative possibilities.

\begin{definition}\label{responsible}
An agent $a\in\mathcal{A}$ is {\bf\em responsible} at a leaf node $v\in L$ if there is a decision node $u\in D$ on the decision path from the root to leaf $v$ 
such that  $u\in win_a(\overline{\ell(v)})$.
\end{definition}


Let us now consider a situation when the United States is attacked by an enemy using nuclear weapons, but the President fails to authorise a retaliatory strike. The decision-making process terminates in node $v_1$, see Figure~\ref{minuteman figure}. Note that the President does not have a strategy to guarantee the strike because the President cannot guarantee that the two keys will be simultaneously turned by officers $A$ and $B$. Thus, the President is not {\em counterfactually} responsible in such a situation\footnote{In outcome $v_1$ the President is responsible for ``seeing-to-it'' that retaliatory strike does not take place. The ``seeing-to-it'' form of responsibility (see \cite{sn25jpl} for an overview) is different from the counterfactual responsibility that we consider in this paper.}. Of course, in this situation, the two officers are not counterfactually responsible either. Formally, by Definition~\ref{win definition},
\begin{equation*}
u_1,v_1\notin 
win_{P}(\textit{Yes})\cup win_{A}(\textit{Yes}) \cup
win_{B}(\textit{Yes}).
\end{equation*}
Hence, in such a situation, none of the agents is responsible for the lack of a retaliatory strike. In other words, there is a {\em responsibility gap}.

The responsibility gap also exists in node $v_2$, where, after the President authorises the launch, at least one of the officers decides not to turn the key. This is because even after the President's authorisation, neither officer has a strategy to guarantee the launch. After all, the other officer can always decide not to turn the key. Formally,
\begin{equation*}
u_1,u_2,v_2\notin 
win_{P}(\textit{Yes})\cup win_{A}(\textit{Yes}) \cup
win_{B}(\textit{Yes}).    
\end{equation*}

\begin{definition}\label{gap-free}
A mechanism is {\bf\em gap-free} if, for each leaf node, there is at least one agent responsible at this node. 
\end{definition}
The mechanism depicted in Figure~\ref{minuteman figure} is {\em not} gap-free because no agent is responsible at leaf nodes $v_1$ and $v_2$.

\section{Elected Dictatorship}

Next, let us turn our attention to a completely different setting inspired by Article I of the US Constitution: {\em ``The Vice President of the United States shall be President of the Senate, but shall have no Vote, unless they be equally divided''}. To make the example more manageable, let us suppose that the Senate contains just two senators: a Republican ($R$) and a Democrat ($D$). If their votes (by a paper ballot) agree, the decision is made. Otherwise, the Vice President breaks the tie, see Figure~\ref{vp figure}.  


\begin{figure}[ht]
    \centering
    \begin{subfigure}[b]{0.22\textwidth}
        \begin{center}
        \scalebox{0.49}{\includegraphics{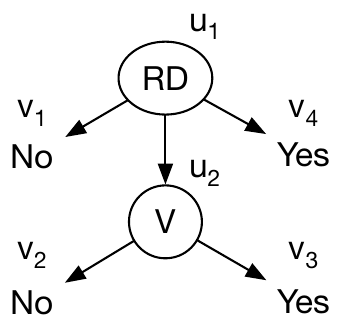}}
        \caption{US Senate mechanism.}
       \label{vp figure}
        \end{center}
    \end{subfigure}
    \begin{subfigure}[b]{0.22\textwidth}
        \begin{center}
        \scalebox{0.49}{\includegraphics{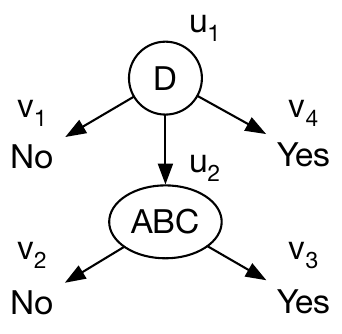}}
        \caption{Academic mechanism.}
        \label{dean figure}
        \end{center}
    \end{subfigure}
    \caption{(a) There is a responsibility gap at nodes $v_1$ and $v_4$. The Vice President is a dictator at node~$u_2$. (b) The mechanism is gap-free. The Dean is a dictator at root node $u_1$.}
\end{figure}

In this setting, the responsibility gap exists in node $v_1$ (outcome ``{\em No}'') and node $v_4$ (outcome ``{\em Yes}''). This is because if the process terminates in either of these nodes, then none of the three agents, at any point during the decision-making process, has a unilateral strategy to prevent the outcome:
\begin{align*}
&u_1,v_1\notin 
win_{R}(\textit{Yes})\cup win_{D}(\textit{Yes})\cup
win_{V}(\textit{Yes}),\\
&u_1,v_4\notin 
win_{R}(\textit{No})\cup win_{D}(\textit{No})\cup
win_{V}(\textit{No}).
\end{align*}

However, the Vice President is counterfactually responsible for the outcome in nodes $v_2$ and $v_3$ because, if the decision-making process reaches either node, then the Vice President has had a chance to prevent the outcome corresponding to the node. In fact, at node $u_2$, the Vice President simultaneously had a strategy to guarantee either of the two possible outcomes:
$u_2\in win_V(\textit{Yes})$ and $u_2\in win_V(\textit{No})$.
We say that the Vice President is a {\em dictator} at node $u_2$. 
\begin{definition}\label{dictator}
An agent $a\in\mathcal{A}$ is a {\bf\em dictator} at a decision node $v\in D$ if $v\in win_a(\text{Yes})$ and $v\in win_a(\text{No})$.
\end{definition}
There is no dictator at any of the nodes in the Two-person Rule mechanism depicted in Figure~\ref{minuteman figure}.

In this paper, we investigate the connection between responsibility gaps and the presence of dictators in a decision mechanism. The example depicted in Figure~\ref{vp figure} shows that the existence of a single dictator is not enough to guarantee that the mechanism is gap-free. However, as we are about to see, the existence of a single dictator condition can be strengthened to provide such a guarantee.

Towards this goal, let us consider one more example. It seems to be a common pattern in academia that administrators prefer to avoid making unpopular decisions by delegating the decision-making to a committee. We capture this situation in the mechanism depicted in Figure~\ref{dean figure}. Here, the Dean ($D$) can either decide on {\em Yes/No} or delegate the decision to a three-member committee consisting of academic staff members $A$, $B$, and $C$. The committee makes the decision by a majority vote using a paper ballot. 

In such a setting, each of the three committee members is never {\em individually} responsible for the outcome because none of them at any moment has an individual strategy that guarantees any of the two outcomes. At the same time, the Dean is not only responsible for the decision in all four leaf nodes, but the Dean is also a dictator at the root node $u_1$.
In particular, this means that there is a dictator at a node on each root-to-leaf decision path of the mechanism. 

\begin{definition}\label{elected dictatorship}
A mechanism is an {\bf\em elected dictatorship} if there is a dictator at one of the decision nodes of each root-to-leaf decision path. 
\end{definition}

The mechanism depicted in Figure~\ref{dean figure} is an extreme example of an elected dictatorship where there is a single dictator at the root node. More generally, different agents might be dictators along different root-to-leaf decision paths of an elected dictatorship. 

It turns out that being an elected dictatorship is not just a sufficient condition for being gap-free, but these two conditions are equivalent. See the theorem below. 

\begin{theorem}\label{big result perfect info}
A mechanism is a gap-free mechanism iff it is an elected dictatorship. 
\end{theorem}

As we show in Appendix~\ref{back to perfect section}, Theorem~\ref{big result perfect info} follows from more general results about games with imperfect information that we establish later in this paper.
This theorem gives a complete characterisation of the ``gap-freeness'' for the class of decision-making mechanisms specified in Definition~\ref{mechanism}. There are two previous related works that considered this property for a much more narrow class of ``discursive dilemma'' mechanisms. In terms of Definition~\ref{mechanism}, a discursive dilemma mechanism is a {\em single-node} mechanism with function $\tau$ having a very special ``criteria-based'' form~\cite{l06ethics}. \citet{dp22scw} considered an alternative best-effort-based definition of responsibility and gave a complete characterisation of gap-free mechanisms in discursive dilemma mechanisms. Their characterisation does not refer to dictatorship. \citet{bh18ej} considered probabilistic discursive dilemma mechanisms and another variation of the definition of responsibility. They proved that if a gap-free discursive dilemma mechanism does not have what they call ``fragmentation'', then the mechanism must be a dictatorship.

\section{Mechanisms with Imperfect Information}

In many decision-making mechanisms, some agents might not have complete information about the actions already taken. An example of such a mechanism is the Drawing Straws\footnote{In 2017, this mechanism was used to decide who gets a seat in the Northumberland County Council (England) after votes were evenly divided~\cite{e17guardian}.}.
Figure~\ref{straw figure} shows a version of this mechanism in which Alice (A) holds both a short and a long straw between her fingers, without revealing which straw is which. Bob (B) picks either left (action 0) or right (action 1) straw. The outcome of the decision-making is determined by whether he has chosen a short or a long straw. The dashed line labelled with $B$ in the figure represents the fact that Bob cannot distinguish nodes $u_2$ and $u_3$ at the moment he chooses an action at those nodes.

\begin{figure}[ht]
\begin{center}
\scalebox{0.49}{\includegraphics{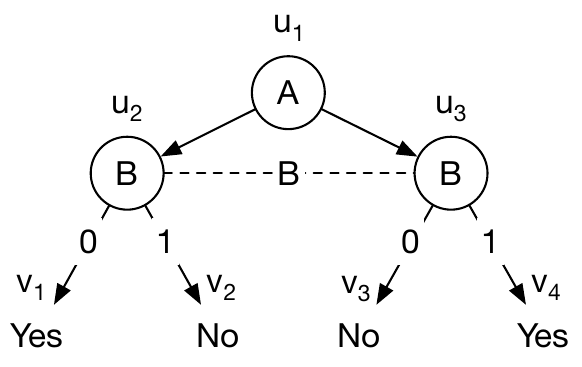}}
\caption{Drawing Straws mechanism.}\label{straw figure}
\end{center}
\end{figure}

To handle the decision-making processes like the one in Figure~\ref{straw figure}, we need a more general notion of a decision-making mechanism in which the agents cannot distinguish some of the decision nodes. 

\begin{definition}\label{imperfect info mechanism}
A decision-making mechanism with imperfect information is a tuple $(V,E,\mathcal{A},\Delta,\tau,\ell,\sim)$ where
\begin{enumerate}
    \item $(V,E,\mathcal{A},\Delta,\tau,\ell)$ is a decision-making mechanism,
    \item $\sim_a$ is an {\bf\em indistinguishability} equivalence relation on the set $D$ of decision nodes for each agent $a\in\mathcal{A}$ such that if $u\sim_a v$, then $\Delta^a_u=\Delta^a_v$.
\end{enumerate}
\end{definition}
Note that the ``if $u\sim_a v$, then $\Delta^a_u=\Delta^a_v$'' requirement of item~2 above specifies that each agent has the same available actions in all indistinguishable nodes. In other words, each agent {\em knows} the actions available to her at the current node.

The set of decision-making mechanisms as specified in Definition~\ref{mechanism} consists of quintuples, while the set of mechanisms with imperfect
information in Definition~\ref{imperfect info mechanism} consists of sextuples. Using the terminology of {\em object-oriented programming} languages, we can say that the class of mechanisms with imperfect information is an {\em extension} of the class of mechanisms from Definition~\ref{mechanism}. Indeed, to treat a mechanism with imperfect information as a mechanism, we just need to ignore the equivalence relations. From this point of view, Definition~\ref{next} through Definition~\ref{elected dictatorship} are still applicable to the mechanisms with imperfect information. Using the object-oriented terminology, one can say that the class of mechanisms with imperfect information {\em inherits} the notions specified in these definitions from the generic class of mechanisms. We adopt such a viewpoint in this paper.

\section{Epistemic Responsibility}

Let us go back to the Drawing Straws mechanism depicted in Figure~\ref{straw figure}. Suppose that  Alice positions long and short straws in such a way that the mechanism transitions from node $u_1$ to node $u_2$. Does Bob have a strategy at node $u_2$ to guarantee the outcome {\em No}? We would say that he does (it is action 1). In fact, it is easy to see that $u_2\in win_B(\textit{No})$ by Definition~\ref{win definition}. Thus, by Definition~\ref{responsible}, Bob is counterfactually responsible at leaf node $v_1$. This observation, however, is not intuitively acceptable: how can Bob be blamed for pulling, say, a long straw if he did not know which of the two straws is long and which is short? This is the reason why in the literature it has been suggested that in order for an agent to be counterfactually responsible in an imperfect information setting, the agent should not only have a strategy to prevent the outcome but also should know what this strategy is~\cite{ydjal19aamas,nt20ai}.

To define what ``to know the strategy'' formally means in our setting is a non-trivial task. In the literature, uniform or ``know-how''~\cite{fhlw17ijcai,nt18ai} strategies are usually defined as functions that assign the same actions to all indistinguishable nodes. Following this approach, we can adjust Definition~\ref{win definition} for the imperfect information setting as shown below. By $[v]_a$ we mean the equivalence class of node $v$ with respect to the equivalence relation $\sim_a$.
\begin{definition}\label{uwin}
For any outcome $o\in  \{\text{Yes},\text{No}\}$ and any agent $a\in\mathcal{A}$, let set 
$uwin_a(o)$ be the smallest subset of $V$ such that, for each node $v\in V$,
\begin{enumerate}
    \item $\ell^{-1}(o)\subseteq uwin_a(o)$,
    \item if $Next^a_d([v]_a)\subseteq uwin_a(o)$, then $v\in uwin_a(o)$, for each decision node $v\in D$ and each action $d\in\Delta^a_v$.
\end{enumerate}
\end{definition}

To define the notion of counterfactual responsibility in decision-making mechanisms with imperfect information, one can consider replacing the set $win_a(\overline{\ell(v)})$ with the set $uwin_a(\overline{\ell(v)})$ in Definition~\ref{responsible}.
Unfortunately, this does {\em not} capture our intuition of what is responsibility in an imperfect information setting. Indeed, consider the decision mechanism with imperfect information depicted in Figure~\ref{uwin figure}.

\begin{figure}[ht]
\begin{center}
\scalebox{0.49}{\includegraphics{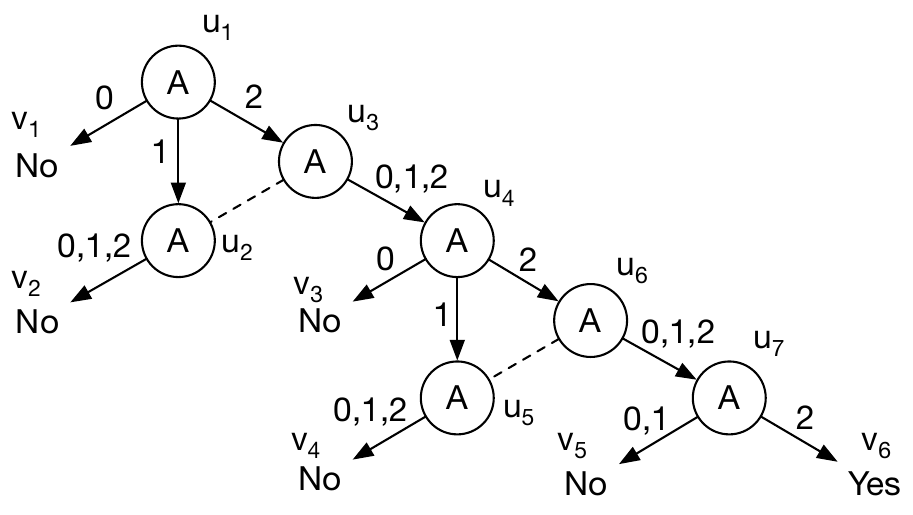}}
\caption{A single-agent decision-making mechanism with imperfect information. Dashed lines represent the relation $\sim_A$.}\label{uwin figure}
\end{center}
\end{figure}
This mechanism has a single agent $A$ (we don't need additional agents to explain the issue). In each decision node, this agent has the same set of actions $\{0,1,2\}$. The outcomes of these actions are shown in the figure. 

Note that $v_6\in uwin_A(\textit{Yes})$ by item~1 of Definition~\ref{uwin}. Also, $u_7\in uwin_A(\textit{Yes})$ by item~2 of Definition~\ref{uwin} (observe that $[u_7]_A=\{u_7\}$ and take $d=2$). At the same time, $u_6\notin uwin_A(\textit{Yes})$. Indeed, $[u_6]_A=\{u_5,u_6\}$ and there is no action $d$ such that $Next^a_d([u_6]_a)\subseteq uwin_a(\textit{Yes})$. As a result, again by Definition~\ref{uwin},
\begin{equation}\label{6-jan-a}
u_4\notin uwin_A(\textit{Yes}).   
\end{equation}
Next, consider a case when the decision process terminates at leaf node $v_3$. The outcome of this process is {\em No}. Could agent $A$ be held counterfactually responsible for this? Intuitively, yes! To reach node $v_3$, the decision path must go through node $u_4$. Agent $A$ can distinguish node $u_4$ from all other nodes in the mechanism, see Figure~\ref{uwin figure}. Thus, while the process is at node $u_4$, agent $A$ knows that the process is at this node. Hence, while at node $u_4$, agent $A$ knows that if she chooses action 2, then she will become ``confused'' (mechanism will transition to node $u_6$ where she will not know how to achieve outcome {\em Yes}). However, in node $u_4$, agent $A$ knows that no matter what she does while being ``confused'' (in node $u_6$), the process will come to a node (in our case $u_7$) where she will wake up from the confusion and will know how to achieve the outcome {\em Yes}.  Having all this information in node $u_4$, agent $A$, we believe, ``knows'' how to achieve outcome {\em Yes} -- take a deep breath and put herself in the state of confusion by choosing action 2. Thus, we think, she should be held counterfactually responsible for the outcome {\em No} in leaf node $v_3$. To achieve this, we need to modify Definition~\ref{uwin} in such a way that statement~\eqref{6-jan-a} is no longer true. 

Of course, one might argue that the issue that we described in the previous paragraph only exists because we allow decision-making mechanisms in which agents can get ``confused''. If we assume that the agents have perfect recall, then the situation depicted in Figure~\ref{uwin figure} will never happen. Specifically, if agent $A$ remembers what action she took in node $u_4$, then she will always be able to distinguish node $u_5$ from node $u_6$. Thus, at node $u_6$ she would still know how to achieve the outcome {\em Yes}.

We agree that Definition~\ref{uwin} {\em seems} to work for agents with perfect recalls and that most commonly used decision-making mechanisms do not force agents to forget what they know. However, there are examples of decision making mechanisms that do so. For instance, a new employee selection mechanism might ask the selection committee to ignore information not shown in the application materials. A judge in court might instruct the jury to ignore certain evidence or a witness testimony. To state our results in the most general form that covers such mechanisms, in this paper we modify Definition~\ref{uwin}. Our revised definition will be able to handle multiple ``confusions/memory losses'' during the decision-making process. For example, consider leaf node $v_1$ in the same mechanism depicted in Figure~\ref{uwin figure}. We believe agent $A$ is counterfactually responsible for the outcome {\em No} at this node. Indeed, at node $u_1$ agent $A$ knows how to achieve outcome {\em Yes} -- by putting herself into the state of confusion twice: first, by taking action 2 at node $u_1$ and then by taking the same action 2 at node $u_4$. Our replacement for Definition~\ref{uwin} is Definition~\ref{ewin} stated below. 

\begin{definition}\label{ewin}
For any outcome $o\in  \{\text{Yes},\text{No}\}$ and any agent $a\in\mathcal{A}$, let set 
$ewin_a(o)$ be the smallest subset of $V$ such that, for each node $v\in V$,
\begin{enumerate}
    \item $\ell^{-1}(o)\subseteq ewin_a(o)$,
    \item if $Next^a_d([v]_a)\subseteq ewin_a(o)$, then $v\in ewin_a(o)$, for each decision node $v\in D$ and each action $d\in\Delta^a_v$,
    \item if $\bigcup_{d\in\Delta^a_v}Next^a_d(v)\subseteq ewin_a(o)$, then $v\in ewin_a(o)$,  for each decision node $v\in D$.
\end{enumerate}
\end{definition}
Note that Definition~\ref{ewin} adds one extra recursive case (item~3) to Definition~\ref{uwin}. This item states that the agent does not need to know how to act in the current node if any possible action in this node leads to the set $ewin_a(o)$. Because Definition~\ref{ewin} adds an {\em extra} case, $uwin_a(o)\subseteq ewin_a(o)$ for each agent $a\in\mathcal{A}$ and each outcome $o\in  \{\textit{Yes},\textit{No}\}$. 

In Figure~\ref{uwin figure}, for example, $v_6\in ewin_A(\textit{Yes})$ by item~1 of Definition~\ref{ewin}. Then, $u_6\in ewin_A(\textit{Yes})$ by item~3 of Definition~\ref{ewin}. Hence, $u_4\in ewin_A(\textit{Yes})$ by item~2 of Definition~\ref{ewin} with $d=2$.






As we discussed after Definition~\ref{imperfect info mechanism}, each mechanism with imperfect information can be viewed as a ``mechanism'' under Definition~\ref{mechanism}. Thus, for the mechanisms with imperfect information, in addition to set $ewin_a(o)$, one can also consider the set $win_a(o)$ as specified in Definition~\ref{win definition}. For example, 
$ewin_B(\textit{Yes})=\{v_1,v_4\}$ and
$win_B(\textit{Yes})=\{v_1,v_4,u_2,u_3\}$ 
 for the Drawing Straw mechanism in Figure~\ref{straw figure}.

The next lemma connects these two sets for an arbitrary mechanism with imperfect information. 

\begin{lemma}\label{ewin sub win}
$ewin_a(o)\subseteq win_a(o)$, for each agent $a\in\mathcal{A}$ and
each outcome $o\in  \{\text{Yes},\text{No}\}$.    
{\em [proof in Appx. A]}
\end{lemma}



We are now ready to define what it means to be counterfactually responsible in decision-making mechanisms with imperfect information. Our definition below simply replaces the set $win_a(\overline{\ell(v)})$ with the set $ewin_a(\overline{\ell(v)})$ in Definition~\ref{responsible}. Since we consider mechanisms with imperfect information to be a ``subclass'' (in the sense of object-oriented programming) of mechanisms, the notion ``responsible'', as specified in  Definition~\ref{responsible} is still technically defined for the mechanisms with imperfect information. To avoid confusion, we use the term ``epistemically responsible'' in the definition below. However, it is important to remember that ``epistemically responsible'' is {\em the} proper definition of being counterfactually responsible in an imperfect information setting. It is the notion that properly captures our intuition about responsibility.

\begin{definition}\label{epistemically responsible}
An agent $a\in\mathcal{A}$ is {\bf\em epistemically responsible} at a leaf node $v\in L$ if there is a decision node $u\in D$ on the decision path from the root to the leaf node $v$ such that $u\in ewin_a(\overline{\ell(v)})$.
\end{definition}
For example, agent $A$ is epistemically responsible at all leaf nodes of the mechanism depicted in Figure~\ref{uwin figure}.

Intuitively, by an ``epistemic gap'' we mean the set of all leaf nodes in which no agent is epistemically responsible.

\begin{definition}\label{epistemic gap-free}
A mechanism is {\bf\em epistemic-gap-free} if, for each leaf node, there is at least one agent epistemically responsible at this node.    
\end{definition}

\section{Elected Epistemic Dictatorship}

The notions of ``dictator at a node'' and ``elected dictatorship'', as specified in Definition~\ref{dictator} and Definition~\ref{elected dictatorship}, have their epistemic counterparts based on function $ewin$ instead of $win$.


\begin{definition}\label{epistemic dictator}
An agent $a\in\mathcal{A}$ is an {\bf\em epistemic dictator} at a node $v$ if $v\in ewin_a(\text{Yes})$ and $v\in ewin_a(\text{No})$.
\end{definition}

\begin{definition}\label{elected epistemic dictatorship}
A mechanism is an {\bf\em elected epistemic dictatorship} if, for each root-to-leaf decision path, there is an epistemic dictator at a node on this path.
\end{definition}

The next theorem shows that the epistemic version of the right-to-left part of Theorem~\ref{big result perfect info} is true.  

\begin{theorem}\label{big result}
Any elected epistemic dictatorship is epistemic-gap-free. 
\end{theorem}

\begin{proof}
Consider any root-to-leaf path $v_1,\dots,v_n$. By Definition~\ref{epistemic gap-free}, it suffices to show that some agent is epistemically responsible at leaf node $v_n$. By the assumption of the theorem that the mechanism is an elected epistemic dictatorship and Definition~\ref{elected epistemic dictatorship}, there is $i<n$ and an agent $a\in\mathcal{A}$ who is an epistemic dictator at decision node $v_i$. Hence, $v_i\in ewin_a(\text{\em Yes})$ and  $v_i\in ewin_a(\text{\em No})$ by Definition~\ref{epistemic dictator}. Thus,
$v_i\in ewin_a(\overline{\ell(v_n)})$. Therefore, agent $a$ is epistemically responsible at leaf node $v_n$ by Definition~\ref{epistemically responsible}.
\end{proof}


Perhaps surprisingly, the epistemic version of the other direction of Theorem~\ref{big result perfect info} is false. To prove this, consider the decision-making mechanism with imperfect information $M$ depicted in Figure~\ref{dic-nogap figure}.

\begin{figure}[ht]
\begin{center}
\scalebox{0.49}{\includegraphics{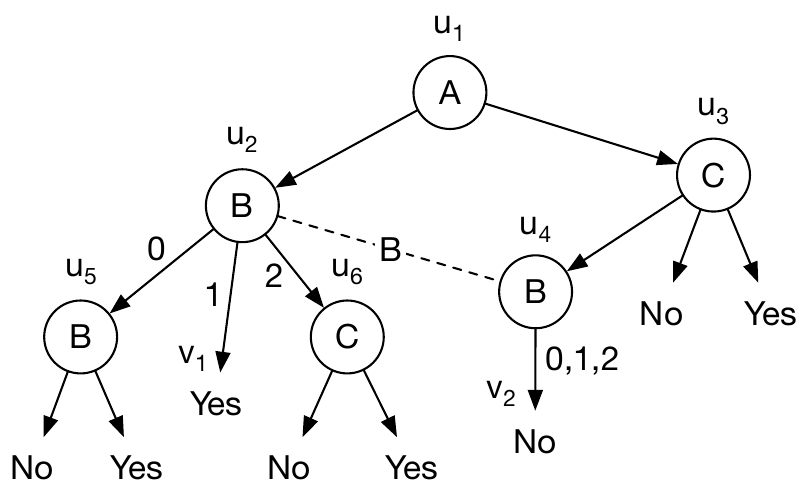}}
\caption{Mechanism $M$ with imperfect information. The names of actions are only shown for the nodes in which the acting agent does not have complete information about the current node.}\label{dic-nogap figure}
\end{center}
\end{figure}

\begin{lemma}\label{M1 lemma}
Mechanism $M$ with imperfect information is epistemic-gap-free.  
{\em [proof in Appx. A]}
\end{lemma}




Note that agent $B$ is a dictator at node $u_2$ because this agent can use action 0 to transition the decision-making process from node $u_2$ to node $u_5\in win_B(\textit{No})$ and this agent can use action 1 to transition the decision-making process from node $u_2$ to node $v_1\in win_B(\textit{Yes})$. However, as we show in the proof of the next lemma, agent $B$ is not an {\em epistemic} dictator at node $u_2$. 

\begin{lemma}\label{M2 lemma}
Mechanism $M$ with imperfect information is not an epistemic elected dictatorship. 
{\em [proof in Appx. A]}
\end{lemma}

Together, Lemma~\ref{M1 lemma} and Lemma~\ref{M2 lemma} show that mechanism $M$ provides a counterexample to the converse of Theorem~\ref{big result}.

In conclusion, observe that mechanism $M$, unlike the mechanism shown in Figure~\ref{uwin figure}, does not make any agent to forget something that she knew before by transitioning the agent into a ``confused'' state. Thus, this mechanism is suitable for agents with perfect recall.
Additionally, observe that the proof of Theorem~\ref{big result} does not rely on the specific definition of the function $ewin$. For example, if the notion of epistemic dictator and epistemic responsibility are defined using function $uwin$ instead of $ewin$, the proof of Theorem~\ref{big result} remains valid. Furthermore, the counterexample given by mechanism $M$ also works for many variations of $ewin$ definition. In particular, it is easy to see that the proofs of Lemma~\ref{M1 lemma} and Lemma~\ref{M2 lemma} remain valid if the notion of epistemic dictator and epistemic responsibility are defined using function $uwin$ instead of $ewin$. These observations show that the results of this section appear to be very general and are not artifacts of the specifics of Definition~\ref{ewin}.

\section{Elected Semi-epistemic Dictatorship}

In Theorem~\ref{big result}, we have shown that, for the mechanisms with imperfect information, the set of elected epistemic dictatorships is a subset of the set of epistemic-gap-free mechanisms. We used the mechanism $M$ to show that the former set is a {\em proper} subset of the latter. Given that Theorem~\ref{big result perfect info} shows that these sets are equal for the mechanism with perfect information, it is natural to ask if the notion of elected epistemic dictatorship could be made {\em weaker} so that the modified set of elected epistemic dictatorships includes the set of all epistemic-gap-free mechanisms. In this section, we propose such a modification: elected semi-epistemic dictatorship.

\begin{definition}\label{semi-epistemic dictator}
An agent $a\in\mathcal{A}$ is a {\bf\em semi-epistemic dictator} at a node $v$ if there exists an outcome $o\in\{\text{Yes},\text{No}\}$ such that 
$v\in ewin_a(o)$ and $v\in win_a(\overline{o})$.
\end{definition}

As an example, note that agent $B$ is a semi-epistemic dictator at node $u_2$ of the mechanism $M$ depicted in Figure~\ref{dic-nogap figure}. Indeed, it is easy to verify that $u_2\in ewin_B(\textit{No})$ and  $u_2\in win_B(\textit{Yes})$.

\begin{definition}\label{semi-epistemic dictatorship}
A mechanism with imperfect information is an {\bf\em elected semi-epistemic dictatorship} if, for each root-to-leaf decision path, there is a semi-epistemic dictator at a node on this path.
\end{definition}

An example of an elected semi-epistemic dictatorship is mechanism $M$ depicted in Figure~\ref{dic-nogap figure}. Indeed, agents $B$ and $C$ are semi-epistemic dictators at nodes $u_2$ and $u_3$, respectively. In fact, $C$ is not just a semi-epistemic dictator, but also an epistemic dictator at node $u_3$.

\begin{figure}[ht]
\begin{center}
\scalebox{0.49}{\includegraphics{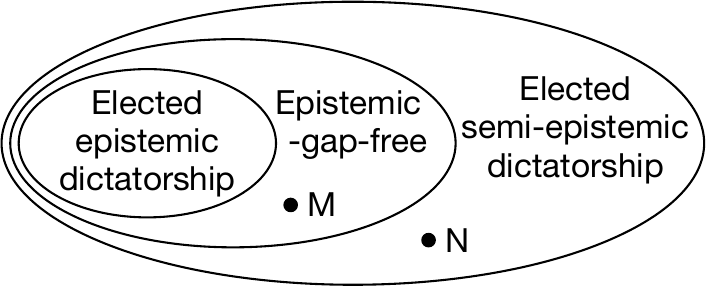}}
\caption{Three sets of decision-making mechanisms with imperfect information. Mechanisms $M$ and $N$ are shown in Figure~\ref{dic-nogap figure} and Figure~\ref{straw give up figure}, respectively.}\label{venn diagram figure}
\end{center}
\end{figure}

In this section, we prove that epistemic-gap-free mechanisms form a subset of elected semi-epistemic dictatorships, see Figure~\ref{venn diagram figure}. We prove this result in Theorem~\ref{big result two}. 



\begin{lemma}\label{step-up-lemma}
For any root-to-node decision path $v_1,\dots,v_k$ of an epistemic-gap-free mechanism, if $v_k\in ewin_a(o)$ and $v_k\notin win_a(\overline{o})$, then there exists $i<k$ and $b\in\mathcal{A}$ such that $v_i\in ewin_b(\overline{o})$.
{\em [proof in Appx. A]}
\end{lemma}


\begin{theorem}\label{big result two}
Any epistemic-gap-free mechanism with imperfect information is an elected semi-epistemic dictatorship. 
\end{theorem}
\begin{proof}
Consider any root-to-leaf decision path $v_1,\dots,v_n$. By Definition~\ref{semi-epistemic dictatorship}, it suffices to show that there is a semi-epistemic dictator at one of the decision nodes of this path. Suppose the opposite. Thus, 
\begin{equation}\label{30-oct-a}
\begin{aligned}
&\text{if $v_i\in ewin_a(o)$, then $v_i\notin win_a(\overline{o})$}\\
&\text{for each $i<n, a\in\mathcal{A},o\in \{\text{Yes},\text{No}\}$}.
\end{aligned}   
\end{equation}
At the same time, by Definition~\ref{epistemic gap-free}, the assumption that the mechanism is epistemic-gap-free implies that there is an agent $b\in\mathcal{A}$ epistemically responsible at leaf node $v_n$. Hence, by Definition~\ref{epistemically responsible}, there is $j<n$ such that $v_j\in ewin_b(\overline{\ell(v_n)})$. Note that $v_j\in ewin_b(\overline{\ell(v_n)})$ implies
$\exists a\exists o (v_j\in ewin_a(o))$.
Thus, $\exists a\exists o (v_j\in ewin_a(o))$.
Then,
there must exist the minimal $j_{\min}\ge 1$ such that
\begin{equation}\label{30-oct-b}
\exists a\exists o (v_{j_{\min}}\in ewin_a(o)).    
\end{equation}
Hence, $v_{j_{\min}}\in ewin_{a'}(o')$ for some agent $a'\in\mathcal{A}$ and some outcome  $o'\in \{\text{Yes},\text{No}\}$. Thus, $v_{j_{\min}}\notin win_{a'}(\overline{o'})$ by statement~\eqref{30-oct-a}.  Then, by Lemma~\ref{step-up-lemma} and the assumption of the theorem that the mechanism is epistemic-gap-free, there exists $j'<j_{\min}$ and an agent $b'\in\mathcal{A}$ such that $v_{j'}\in ewin_{b'}(\overline{o'})$. The last statement contradicts the choice of $j_{\min}$ as the minimal one satisfying condition~\eqref{30-oct-b}.    
\end{proof}

In Theorem~\ref{big result two}, we have shown that the set of epistemic-gap-free mechanisms is a subset of the set of elected semi-epistemic dictatorships. This subset is {\em proper}, see Figure~\ref{venn diagram figure}. To prove this, let us modify the Drawing Straws mechanism depicted in Figure~\ref{straw figure}. Recall that this mechanism is sometimes used to determine the outcome of an election when votes are evenly divided. Suppose that outcome {\em Yes} means that Bob (and not Alice) becomes an elected official. We modify the Drawing Straws mechanism by giving Bob an option to give up the race and let Alice to become the elected official. This option is represented by action 2 in Figure~\ref{straw give up figure} showing the modified mechanism $N$.

\begin{figure}[ht]
\begin{center}
\scalebox{0.49}{\includegraphics{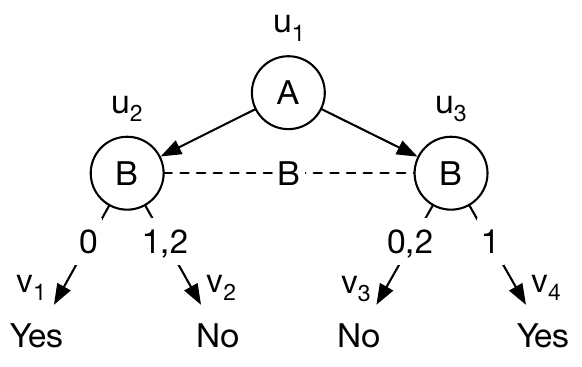}}
\footnotesize
\caption{\footnotesize Mechanism $N$ with imperfect information. The names of actions are only shown for the nodes in which the acting agent does not have complete information about the current node.}\label{straw give up figure}
\end{center}
\end{figure}

\begin{lemma}\label{N1 lemma}
Mechanism $N$ with imperfect information is an elected semi-epistemic dictatorship. {\em [proof in Appx. A]}
\end{lemma}

\begin{lemma}\label{N2 lemma}
Mechanism $N$ with imperfect information is not epistemic-gap-free.
{\em [proof in Appx. A]}
\end{lemma}
Together, the last two lemmas show that the set of epistemic-gap-free mechanisms is a {\em proper} subset of the set of elected semi-epistemic dictatorships, see Figure~\ref{venn diagram figure}.

\section{Conclusion}

This paper contains two main results. First, in the perfect information case, the only way to avoid a responsibility gap in a decision-making mechanism is to use an ``elected dictatorship'', where the agents agree on a single person who makes the decision. The converse is also true: any elected dictatorship is gap-free. Second, in the imperfect information case, the situation is more complicated: epistemic-gap-free mechanisms are ``squeezed'' between elected epistemic and elected semi-epistemic dictatorships. Intuitively, our results mean that to construct non-dictatorial gap-free mechanisms, one needs to consider other forms of responsibility.

\begin{figure}[ht]
\begin{center}
\scalebox{0.49}{\includegraphics{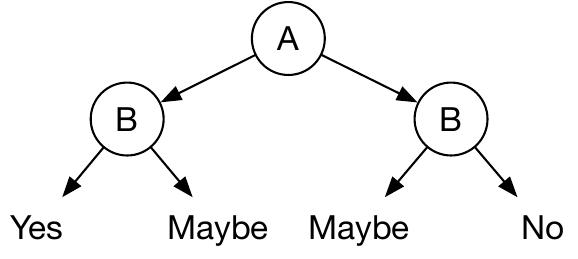}}
\caption{Gap-free mechanism with three alternatives which is not an elected dictatorship.}\label{three-alt figure}
\end{center}
\end{figure}

Note that all our results assume that the decision is binary ({\em Yes/No}). If the third (``{\em Maybe}'') alternative is added, then our results are no longer true. For example, Figure~\ref{three-alt figure} depicts a gap-free mechanism (agent $B$ can prevent any specific outcome) in which none of the agents is a dictator (has a strategy to guarantee any outcome) at any of the nodes.

\bibliographystyle{named}
\bibliography{naumov}

\clearpage

\begin{center}
\Large\sc Technical Appendix    
\end{center}

The lemmas 1 through 6 have been stated in the main part of this. As a result of this, they have smaller numbers.

\appendix

\section{Auxiliary Lemmas}\label{Proofs of Lemmas section}

The next lemma follows from the contraposition of item~2 in Definition~\ref{win definition}. This lemma is used in the proof of Lemma~\ref{step-down-lemma}. 
\begin{lemma}\label{next-exists-lemma}
For any decision node $v\notin win_a(o)$ and any action $d\in\Delta^a_v$, we have $Next^a_d(v)\nsubseteq win_a(o)$. 
\end{lemma}

The following two lemmas are directly used in the proof of Lemma~\ref{step-up-lemma}.
\begin{lemma}\label{two hares} 
$win_a(o)\cap win_b(\overline{o})=\varnothing$ for each distinct agents $a,b\in \mathcal{A}$ and each outcome $o\in  \{\text{Yes},\text{No}\}$.
\end{lemma}
\begin{proof}
Suppose that the set $win_a(o)\cap win_b(\overline{o})$ is not empty. Consider a node $v$ of this set with the minimal height. 

If $v$ is a leaf node, then by Definition~\ref{win definition}, the statements $v\in win_a(o)$ and $v\in win_b(\overline{o})$ imply that $\ell(v)=o$ and $\ell(v)=\overline{o}$, respectively, which is a contradiction.

If $v$ is a decision node, then by Definition~\ref{win definition}, the statements $v\in win_a(o)$ and $v\in win_b(\overline{o})$ imply that there are actions $d_a\in \Delta^a_v$ and $d_b\in \Delta^b_v$ such that
\begin{equation}\label{7-jan-a}
Next^a_{d_a}(v)\subseteq win_a(o)  
\text{\;and\; }
Next^b_{d_b}(v)\subseteq win_b(\overline{o}).
\end{equation}
Consider an arbitrary action profile $\delta$ at node $v$ such that $\delta_a=d_a$ and $\delta_b=d_b$. Such a profile exists because agents $a$ and $b$ are distinct by the assumption of the lemma and set $\Delta^c_v$ is not empty for each agent $c\in\mathcal{A}$ by item~3 of Definition~\ref{mechanism}. 
Thus, $\tau_v(\delta)\in Next^a_{d_a}(v)$ and $\tau_v(\delta)\in Next^b_{d_b}(v)$ by Definition~\ref{next}.
Hence, by item~2 of Definition~\ref{win definition},
\begin{equation}\label{7-jan-b}
\tau_v(\delta)\in win_a(o)
\text{\;\; and\;\; }
\tau_v(\delta)\in win_b(\overline{o}).    
\end{equation}
Note that node $\tau_v(\delta)\in Ch_v$ by item~4 of Definition~\ref{mechanism}. Thus, the node $\tau_v(\delta)$ has a smaller height than node $v$. Therefore, statement~\eqref{7-jan-b} contradicts the choice of node $v$ as the smallest height node in the set $win_a(o)\cap win_b(\overline{o})$. 
\end{proof}

\begin{lemma}\label{base-lemma}
For any agent $a\in\mathcal{A}$,
any outcome $o\in  \{\text{Yes},\text{No}\}$, and any leaf node $v$, if $v\in ewin_a(o)$, then $\ell(v)=o$.  
\end{lemma}
\begin{proof}
The statement of the lemma follows from  Definition~\ref{ewin} because, by Definition~\ref{next}, function $Next^a_d(v)$ is defined only for decision (non-leaf) nodes $v$.     
\end{proof}

The following lemma is used directly in the proofs of Lemma~\ref{ewin sub win} and Lemma~\ref{step-down-lemma}, both of which are then used in the proof of Lemma~\ref{step-up-lemma}.
\begin{lemma}\label{next-all-lemma}
For any agent $a\in\mathcal{A}$,
any outcome $o\in  \{\text{Yes},\text{No}\}$, and 
any decision node $v\in ewin_a(o)$, there exists an action $d\in\Delta^a_v$ such that $Next^a_d(v)\subseteq ewin_a(o)$.
\end{lemma}
\begin{proof}
By Definition~\ref{ewin}, the assumption $v\in ewin_a(o)$ implies that one of the following cases takes place:

\noindent{\em Case 1}: $Next^a_d([v]_a)\subseteq ewin_a(o)$ for some action $d\in \Delta^a_v$. Then, $Next^a_d(v)\subseteq ewin_a(o)$.  

\noindent{\em Case 2}:
$\bigcup_{d\in\Delta^a_v}Next^a_d(v)\subseteq ewin_a(o)$. Note that set $\Delta^a_v$ contains at least one action $d$ by item~3 of Definition~\ref{mechanism}. Then,
$Next^a_d(v)\subseteq ewin_a(o)$.
\end{proof}

\noindent{\bf Lemma~\ref{ewin sub win}.} {\em
$ewin_a(o)\subseteq win_a(o)$, for each agent $a\in\mathcal{A}$ and
each outcome $o\in  \{\text{Yes},\text{No}\}$.   
}
\begin{proof}
It suffices to show that $v\in win_a(o)$ for each node $v\in ewin_a(o)$. We prove this statement by induction on the height of node $v$ in the mechanism tree.

If $v$ is a leaf node, then the assumption $v\in ewin_a(o)$ implies $v\in \ell^{-1}(o)$ by Definition~\ref{ewin}. Therefore, $v\in win_a(o)$ by item~1 of Definition~\ref{win definition}. 

If $v$ is a decision node, then, by Lemma~\ref{next-all-lemma}, there exists an action $d\in\Delta^a_v$ such that $Next^a_d(v)\subseteq ewin_a(o)$. By Definition~\ref{next}, all nodes in the set $Next^a_d(v)$ have smaller height than node $v$. Hence, by the induction hypothesis, $Next^a_d(v)\subseteq win_a(o)$. Therefore, $v\in win_a(o)$ by item~2 of Definition~\ref{win definition}.
\end{proof}

\noindent{\bf Lemma~\ref{M1 lemma}.} {\em
Mechanism with imperfect information $M$ is epistemic-gap-free.    
}
\begin{proof}
By Definition~\ref{epistemic gap-free}, it suffices to show that for each leaf node of mechanism $M$ there is at least one agent epistemically responsible at that node.

First, note that $u_3\in ewin_C(\text{\em Yes})$ by Definition~\ref{ewin}. Thus, agent $C$ is epistemically responsible at leaf node $v_2$.

Second, observe that $u_5,v_2\in ewin_B(\text{No})$. Thus, $Next^B_{0}([u_2]_B)=Next^B_{0}(\{u_2,u_4\})=\{u_5,v_2\}\subseteq ewin_B(\text{\em No})$. Hence, $u_2\in ewin_B(\text{\em No})$ by item~2 of Definition~\ref{ewin}. Then, agent $B$ is epistemically responsible at leaf node $v_1$.

Epistemic responsibility at the other six leaf nodes of mechanism $M$ is straightforward.
\end{proof}

\noindent{\bf Lemma~\ref{M2 lemma}.} {\em
Mechanism $M$ is not an epistemic elected dictatorship. 
}
\begin{proof}
By Definition~\ref{elected epistemic dictatorship}, it suffices to show that there is no epistemic dictator at any of the three nodes of the decision path from root node $u_1$ to leaf node $v_1$. We consider these three nodes separately.

\vspace{1mm}
\noindent{\bf Node $u_1$}. Agent $A$ has no strategy at node $u_1$ to unilaterally guarantee either outcome {\em Yes} or {\em No}. Indeed, no matter if agent $A$ chooses node $u_2$ or $u_3$, depending on the actions of agents $B$ and $C$ that follow, the outcome can be either {\em Yes} or {\em No}. Thus, in particular, $u_1\notin win_A(\textit{Yes})$. Hence, $u_1\notin ewin_A(\textit{Yes})$ by Lemma~\ref{ewin sub win}. Therefore, agent $A$ is not an epistemic dictator at node $u_1$ by Definition~\ref{epistemic dictator}.

Agent $B$ also has no strategy at node $u_1$ to unilaterally guarantee either outcome {\em Yes} or {\em No}. This is because if agent $A$ chooses node $u_3$, then it is up to agent $C$ to decide on the outcome. Thus, in particular, $u_1\notin win_B(\textit{Yes})$. Hence, $u_1\notin ewin_B(\textit{Yes})$ by Lemma~\ref{ewin sub win}. Therefore, agent $B$ is not an epistemic dictator at node $u_1$ by Definition~\ref{epistemic dictator}. 

Similarly, if at node $u_1$ agent $A$ chooses $u_2$, then agent $C$ has no strategy to guarantee either of the two outcomes. Therefore, agent $C$ is not an epistemic dictator at node $u_1$. 

\vspace{1mm}
\noindent{\bf Node $u_2$}. It is easy to see that neither agent $A$ nor agent $C$ has a unilateral strategy at node $u_2$ to guarantee either of the two outcomes. Thus, using an argument similar to the one above, we can conclude that neither of these two agents is an epistemic dictator at node $u_2$. 

Towards a contradiction, suppose that agent $B$ is an epistemic dictator at node $u_2$. Thus, $u_2\in ewin_B(\textit{Yes})$ by Definition~\ref{epistemic dictator}. Hence, by Definition~\ref{ewin}, one of the following cases must take place:

\noindent{\em Case 1}: $Next^a_d([u_2]_B)\subseteq ewin_B(\textit{Yes})$ for some action $d\in \{0,1,2\}$ of agent $B$ at node $u_2$.
Thus, $Next^a_d(u_4)\subseteq ewin_B(\textit{Yes})$ because $u_2\sim_B u_4$, see Figure~\ref{dic-nogap figure}.
Hence, $v_2\in ewin_B(\textit{Yes})$, see again Figure~\ref{dic-nogap figure}. Therefore, $\ell(v_2)=\textit{Yes}$ by Definition~\ref{ewin}, which is a contradiction, see Figure~\ref{dic-nogap figure}.

\noindent{\em Case 2}: $\bigcup_{d\in\Delta^B_{u_2}}Next^B_d(u_2)\subseteq ewin_B(\textit{Yes})$. Then, $Next^B_2(u_2)\subseteq ewin_B(\textit{Yes})$. Thus, $u_6\in ewin_B(\textit{Yes})$. Hence, $u_6\in win_B(\textit{Yes})$ by Lemma~\ref{ewin sub win}, which is a contradiction because agent $B$ has no strategy to guarantee outcome $\textit{Yes}$ at node $u_6$.

\vspace{1mm}
\noindent{\bf Node $v_1$}. Note that $\ell(v_1)=\textit{Yes}$. Thus, $v_1\notin ewin_X(\textit{No})$ for each agent $X\in \{A,B,C\}$. Hence, by Definition~\ref{epistemic dictator}, none of the three agents is an epistemic dictator at node $v_1$.
\end{proof}

\begin{lemma}\label{step-down-lemma}
For any decision node $v\in D$, if $v\in ewin_a(o)$ and $v\notin win_a(\overline{o})$, then there is a node $u\in Ch_v$ such that $u\in ewin_a(o)$ and $u\notin win_a(\overline{o})$.   
\end{lemma}
\begin{proof}
By Lemma~\ref{next-all-lemma}, the assumption  $v\in ewin_a(o)$ implies that there is an action $d\in\Delta^a_v$ such that 
\begin{equation}\label{16-oct-h}
Next^a_d(v)\subseteq ewin_a(o).    
\end{equation}

Also, by Lemma~\ref{next-exists-lemma}, the assumption $v\notin win_a(\overline{o})$ implies $Next^a_d(v)\nsubseteq win_a(\overline{o})$. Then,  there is a node 
\begin{equation}\label{16-oct-i}
u\in Next^a_d(v)    
\end{equation}
such that $u\notin win_a(\overline{o})$. Note that  $u\in ewin_a(o)$ by statement~\eqref{16-oct-i} and statement~\eqref{16-oct-h}. Finally, statement~\eqref{16-oct-i} implies $u\in Ch_v$ by Definition~\ref{next} and item~4 of Definition~\ref{mechanism}. 
\end{proof}

\noindent{\bf Lemma~\ref{step-up-lemma}.} {\em
For any root-to-node decision path $v_1,\dots,v_k$ of an epistemic-gap-free mechanism, if $v_k\in ewin_a(o)$ and $v_k\notin win_a(\overline{o})$, then there exists $i<k$ and $b\in\mathcal{A}$ such that $v_i\in ewin_b(\overline{o})$. 
}
\begin{proof}
By applying Lemma~\ref{step-down-lemma} multiple number of times, root-to-node decision path $v_1,\dots,v_k$ can be extended to a root-to-{\bf\em leaf} decision path $v_1,\dots,v_k,v_{k+1},\dots,v_n$ such that
\begin{align}
v_j\in ewin_a(o)    &\text{ for each $j\ge k$},\label{29-oct-a}\\
v_j\notin win_a(\overline{o}) &\text{ for each $j\ge k$}.\label{29-oct-b}
\end{align}
By Lemma~\ref{ewin sub win}, statement~\eqref{29-oct-a} implies that $v_j\in win_a(o)$ for each $j\ge k$. Then, $v_j\notin win_{c}(\overline{o})$ for each $j\ge k$ and each agent $c\neq a$ by Lemma~\ref{two hares}. Hence, 
$v_j\notin win_c(\overline{o})$ for each $j\ge k$ and each agent $c\in\mathcal{A}$ 
by statement~\eqref{29-oct-b}. Thus, by Lemma~\ref{ewin sub win},
\begin{equation}\label{29-oct-c} 
v_j\notin ewin_c(\overline{o}) \text{ for each $j\ge k$ and each agent $c\in\mathcal{A}$}.   
\end{equation}
In addition, statement~\eqref{29-oct-a} implies that $v_n\in ewin_a(o)$. Thus, by Lemma~\ref{base-lemma} and because $v_n$ is a leaf node,
\begin{equation}\label{29-oct-d} 
\ell(v_n)=o.
\end{equation}

By Definition~\ref{epistemic gap-free}, the assumption of the lemma that the mechanism is epistemic gap-free implies that there is an agent $b$ epistemically responsible at leaf node $v_n$. Thus, by Definition~\ref{responsible}, there is $i<n$ such that 
$v_i\in ewin_b(\overline{\ell(v_n)})$. Hence, $v_i\in ewin_b(\overline{o})$ by equation~\eqref{29-oct-d}. Therefore, $i<k$ by statement~\eqref{29-oct-c}.
\end{proof}

\noindent{\bf Lemma~\ref{N1 lemma}.} {\em
Mechanism with imperfect information $N$ is an elected semi-epistemic dictatorship.
}
\begin{proof}
Note that $v_1\in win_B(\textit{Yes})$ and $v_2,v_3\in ewin_B(\textit{No})$ by item~1 of Definition~\ref{win definition} and item~1 of Definition~\ref{ewin}, respectively, see Figure~\ref{straw give up figure}. Thus,
$Next^B_0(u_2)=v_1\in win_B(\textit{Yes})$ and 
$Next^B_2([u_2]_B)=Next^B_2(\{u_2,u_3\})=\{v_2,v_3\}\subseteq ewin_B(\textit{No})$ again see Figure~\ref{straw give up figure}. Hence, $u_2\in win_B(\textit{Yes})$ and  $u_2\in ewin_B(\textit{No})$
by item~2 of Definition~\ref{win definition} and item~2 of Definition~\ref{ewin}, respectively. Therefore, agent $B$ is a semi-epistemic dictator at node $u_2$.
It can be similarly shown that $B$ is also a semi-epistemic dictator at node $u_3$. Therefore, $N$ is an elected semi-epistemic dictatorship by Definition~\ref{semi-epistemic dictator}.
\end{proof}

\noindent{\bf Lemma~\ref{N2 lemma}.} {\em
Mechanism $N$ with imperfect information is not epistemic-gap-free.
}
\begin{proof}
By Definition~\ref{ewin}, 
$
ewin_A(\text{\em Yes})=ewin_B(\text{\em Yes})=\{v_1,v_4\}   
$.
Thus, $u_1,u_2,v_2\notin ewin_A(\textit{Yes})$ and $u_1,u_2,v_2\notin ewin_B(\textit{Yes})$. Hence, neither of the two agents is epistemically responsible at node $v_2$ by Definition~\ref{epistemically responsible}. Therefore, $N$ is not epistemic-gap-free by Definition~\ref{epistemic gap-free}.
\end{proof}

\section{Proof of Theorem~\ref{big result perfect info}}\label{back to perfect section}

In this section, we prove Theorem~\ref{big result perfect info}. Recall that, after Definition~\ref{imperfect info mechanism} we discussed how any mechanism with imperfect information can be treated as a mechanism (as in Definition~\ref{mechanism}) by ignoring the indistinguishability relation. Let us now observe that the opposite is true as well:
any decision-making mechanism satisfying Definition~\ref{mechanism} can be viewed as a special case of decision-making mechanisms with imperfect information by assuming trivial equivalence relations on nodes (each decision node is indistinguishable only from itself by each agent). This way, we are able to use results about decision-making mechanisms with imperfect information to prove Theorem~\ref{big result perfect info}.

\vspace{1mm}
\noindent{\bf Theorem~\ref{big result perfect info}.} {\em
A a mechanism is gap-free mechanism iff it is an elected dictatorship. 
}
\begin{proof}
Consider any mechanism $M$ specified as in Definition~\ref{mechanism}. When this mechanism is viewed as a mechanism with imperfect information, we have
\begin{equation}\label{13-jan-a}
\forall a\in\mathcal{A}\, \forall o\in\{\textit{Yes},\textit{No}\}\, (ewin_a(o)=win_a(o))   
\end{equation}
by Definition~\ref{ewin} and Definition~\ref{win definition}. The statement of the theorem follows from the following two claim.
\begin{claim}
$M$ is gap-free iff it is epistemic-gap-free. 
\end{claim}
\begin{proof-of-claim}
By Definition~\ref{responsible} and Definition~\ref{epistemically responsible}, statement~\eqref{13-jan-a} implies that an agent is responsible at some leaf node iff the agent is {\em epistemically} responsible at the node. Therefore, by Definition~\ref{gap-free} and Definition~\ref{epistemic gap-free}, mechanism $M$ is gap-free iff it is epistemic-gap-free.
\end{proof-of-claim}

\begin{claim}
Mechanism $M$ is epistemic-gap-free iff it is an elected dictatorship.
\end{claim}
\begin{proof-of-claim}
By Definition~\ref{dictator}, Definition~\ref{epistemic dictator} and Definition~\ref{semi-epistemic dictator}, statement~\eqref{13-jan-a} implies that for any agent $a\in\mathcal{A}$ and any decision node $v\in D$, the following three statements are equivalent:
\begin{enumerate}
\item agent $a$ is a dictator at $v$, 
\item agent $a$ is a semi-epistemic dictator at $v$,
\item agent $a$ is an epistemic dictator at $v$.
\end{enumerate}
Thus, by Definition~\ref{elected dictatorship}, Definition~\ref{semi-epistemic dictatorship}, and Definition~\ref{elected epistemic dictatorship}, the following three statements are equivalent:
\begin{enumerate}
    \item $M$ is an elected dictatorship,
    \item $M$ is an elected semi-epistemic dictatorship,
    \item $M$ mechanism is an elected epistemic dictatorship.
\end{enumerate}
Then, $M$ is epistemic-gap-free iff $M$ is an elected dictatorship by Theorem~\ref{big result} and Theorem~\ref{big result two} (also see Figure~\ref{venn diagram figure}).
\end{proof-of-claim}
This concludes the proof of the theorem.
\end{proof}

\end{document}